\documentclass[11pt,a4]{article}
\pdfoutput=1

\usepackage{microtype}

\usepackage[T1]{fontenc}
\usepackage{authblk}
\usepackage{amsthm}
\usepackage{amssymb}
\usepackage[english]{babel}
\usepackage[utf8]{inputenc}
\usepackage{amsmath}
\usepackage{amsfonts}
\usepackage{mathtools}
\usepackage{enumitem}
\usepackage{thmtools}
\usepackage{thm-restate}
\usepackage{fullpage}

\usepackage{clrscode3e}
\usepackage[noend, noline, ruled, linesnumbered]{algorithm2e}
\SetAlFnt{\normalsize}
\DontPrintSemicolon
\SetKwComment{Comment}{$\triangleright$\ }{}
\SetKwProg{Def}{def}{:}{}
\SetArgSty{textnormal}
\SetKwRepeat{Do}{do}{while}
\SetKwComment{Comment}{$\triangleright$\ }{}

\usepackage{color}
\definecolor{darkgreen}{rgb}{0,0.5,0}
\definecolor{darkblue}{rgb}{0,0,0.8}
\definecolor{darkred}{rgb}{0.8,0,0}

\usepackage{hyperref}
\hypersetup{
   unicode=false,          
   colorlinks=true,        
   linkcolor=darkred,          
   citecolor=darkblue,        
   filecolor=magenta,      
   urlcolor=black           
}

\newcommand{\bigo}{\mathcal{O}}
\newcommand{\PhiS}{\Phi^{*}}
\newcommand{\LambdaS}{\Lambda^{*}}
\newcommand{\sample}{S}

\newcommand{\vweight}{\omega}
\newcommand{\liecount}[1]{\ell_{#1}}

\newcommand{\target}{v^{*}}
\newcommand{\E}{\mathbb{E}}
\newcommand{\T}{\tau}

\theoremstyle{plain}
\newtheorem{definition}{Definition}[section]
\newtheorem{lemma}[definition]{Lemma}
\newtheorem{theorem}[definition]{Theorem}

\newtheorem{observation}[definition]{Observation}
\newtheorem{corollary}[definition]{Corollary}

\newtheorem{lemma-restate}[definition]{Lemma}

\title{\bf An Efficient Noisy Binary Search in Graphs via Median Approximation}

\author[1]{\Large Dariusz Dereniowski\thanks{Partially supported by National Science Centre (Poland) grant number 2018/31/B/ST6/00820.}}
\author[2]{Aleksander \L{}ukasiewicz}
\author[2]{Przemys\l{}aw~Uzna\'nski}
\affil[1]{\small Faculty of Electronics, Telecommunications and Informatics, Gda\'{n}sk~University~of~Technology,~Poland}
\affil[2]{Institute of Computer Science, University~of~Wroc\l{}aw,~Poland}
\date{}

\begin{document}

\maketitle

\setcounter{page}{0}
\thispagestyle{empty}

\begin{abstract}
Consider a generalization of the classical binary search problem in linearly sorted data to the graph-theoretic setting.
The goal is to design an adaptive query algorithm, called a \emph{strategy}, that identifies an initially unknown \emph{target} vertex in a graph by asking queries.
Each query is conducted as follows: the strategy selects a vertex $q$ and receives a reply $v$: if $q$ is the target, then $v=q$, and if $q$ is not the target, then $v$ is a neighbor of $q$ that lies on a shortest path to the target.
Furthermore, there is a noise parameter $0\leq p<\frac{1}{2}$, which means that each reply can be incorrect with probability $p$.
The optimization criterion to be minimized is the overall number of queries asked by the strategy, called the \emph{query complexity}.
The query complexity is well understood to be $\bigo(\varepsilon^{-2}\log n)$ for general graphs, where $n$ is the order of the graph and $\varepsilon=\frac{1}{2}-p$.
However, implementing such a strategy is computationally expensive, with each query requiring possibly $\bigo(n^2)$ operations.

In this work we propose two efficient strategies that keep the optimal query complexity.
The first strategy achieves the overall complexity of $\bigo(\varepsilon^{-1}n\log n)$ per a single query.
The second strategy is dedicated to graphs of small diameter $D$ and maximum degree $\Delta$ and has the average complexity of $\bigo(n+\varepsilon^{-2}D\Delta\log n)$ per query.
We stress out that we develop an algorithmic tool of graph median approximation that is of independent interest: the median can be efficiently approximated by finding a vertex minimizing the sum of distances to a randomly sampled vertex subset of size $\bigo(\varepsilon^{-2}\log n)$.
\end{abstract}

\newpage 
\setcounter{page}{1}

\section{Introduction} \label{sec:intro}

Our research problems originate in the classical ``twenty questions game'' proposed by R\'{e}nyi~\cite{Renyi61} and Ulam~\cite{Ulam76}.
The classical problem of binary search with erroneous comparisons received a considerable attention and optimal query complexity algorithms are known, see e.g. \cite{Ben-OrH08,BorgstromK93,DhagatGW92,FeigeRPU94,RivestMKWS80} for asymptotically best results.
The binary search in linearly ordered data can be re-casted as a search on a path, where each query selects a vertex $q$ and reply gives whether the target element is $q$, or is to the left or to the right of $q$.
This leads to the graph search problem introduced first for trees by Onak and Parys in \cite{OnakP06} and then recently for general graphs by Emamjomeh-Zadeh et~al. in \cite{Emamjomeh-Zadeh:2015aa}.
We recall a following formal statement.

\setlist[description]{leftmargin=\parindent,labelindent=\parindent}
\begin{description}
 \item[Problem formulation.] Consider an arbitrary simple graph $G$ whose one vertex $\target$ is marked as the \emph{target}.
The target is unknown to the query algorithm.
Each query points to a vertex $q$, and a \emph{correct} reply does the following: if $\target=q$, then the reply returns $q$, and if $\target\neq q$, then the reply returns a neighbor of $q$ that belongs to a shortest path from $q$ to $\target$, breaking ties arbitrarily.
We further assume that some replies can be incorrect: each query receives an erroneous reply (independently) with some fixed probability $0\leq p<\frac{1}{2}$ (the value of the \emph{noise parameter} $p$ is known to the algorithm). 
The goal is to design an algorithm, also called a \emph{strategy} performing as few queries as possible.
\end{description}

Typically in the applications of the adaptive query problems the main concern is the number of queries to be performed, i.e., their \emph{query complexity}.
This is due to the fact that the queries usually model a time consuming and complex event like making a software check to verify whether it contains a malfunctioning piece of code, c.f. Ben-Asher et al. \cite{Ben-AsherFN99}, or asking users for some sort of feedback c.f. Emamjomeh-Zadeh and Kempe \cite{Emamjomeh-ZadehK17}.
However, as a second measure the computational complexity comes into play and it is of practical interest to resolve the question of having an adaptive query algorithm that keeps an optimal query complexity and optimizes the computational cost as a second criterion.
This may be especially useful in cases when queries are fast, like communication events over a noisy channel.

The asymptotics of the query complexity is quite well understood to be roughly $\frac{\log n}{1-H(p)}=\bigo(\varepsilon^{-2}\log n)$ (c.f. \cite{DereniowskiTUW19,Emamjomeh-Zadeh:2015aa}), where $n$ is the order of the graph, $\varepsilon=\frac{1}{2}-p$, and $H(p)=-p\log_2 p-(1-p)\log_2(1-p)$ is the entropy.
Thus, it is of theoretical and practical interest to know what is the optimal complexity of computing each particular query.
This leads us to a general statement of the type of solution we seek.
\begin{description}
 \item[Research question.]
 How much the computational complexity of an adaptive graph query algorithm can be improved without worsening the query complexity?
\end{description}
In this work we make the following assumption: a \emph{distance oracle} is available to the algorithm and it gives the graph distance between any pair of vertices. 
This is dictated by the observation that the computation of multiple-pair shortest paths throughout the search would dominate the computational complexity.
On the other hand, we note that this is only used to resolve (multiple times) the following for a query: given a vertex $q$, its neighbor $v$ and an arbitrary vertex $u$, does $v$ lie on a shortest path from $q$ to $u$?
Thus, some weaker oracles can be assumed instead.
We further comment on this assumption in the next section.

\subsection{Motivation} \label{sec:motivation}

To sketch potential practical scenarios of using graph queries we mention a set of examples given in \cite{Emamjomeh-ZadehK17}.
These examples are anchored in the field of machine learning, and since they have the same flavor with respect as how graphs are used, we refer to one of them.
Consider a situation in which a \emph{system} wants to learn a clustering by asking queries.
Each query presents a potential clustering to a user and if this is not the target clustering, then as a response the user either points two clusters that should be merged or points one cluster that should be split (but does not say how to split it).
Thus, the goal is to construct a query algorithm to be used by the system.
It turns out that learning the clustering can be done by asking queries on a graph: each vertex $v$ corresponds to a clustering and a reply of the user for $v$ will be aligned with one of the edges incident to $v$.
In other words, the reply can be associated with an edge outgoing from $v$ that lies on a shortest path to the desired target clustering.
We emphasize some properties of this approach.
First, the fact that the reply indeed reveals the shortest path to the target is an important property of the underlying graph used by the algorithm and thus the graph needs to be carefully defined to satisfy it.
Second, the user is not aware of the fact that such a graph-theoretic approach is used, as only a series of proposed clustering is presented.
Third, this approach is resilient to errors on the user side: the graph query algorithms easily handle the facts that some replies can be incorrect (the user may make a mistake, or may not be willing to reveal the truth).
It has been shown \cite{Emamjomeh-ZadehK17} that in a similar way one can approach the problems of learning a classifier or learning a ranking.

From the standpoint of complexity we can approach such scenarios in two ways.
First, one can derive an algorithm that specifically targets a particular application.
More precisely, if one considers one of the above applications, then it may turn out that e.g. it is not necessary to construct the entire graph but instead reconstruct only what is necessary to perform each query.
The second way is the general approach taken in this work: to consider the underlying graph as an abstract data structure out of the context of how it is used in particular applications.
We note that examples like the ones mentioned above reveal that some applications may be burdened by the fact that the underlying graph is large, in which case the computational complexity, or local search procedures may be more crucial.

We finally comment on our assumption that a shortest path oracle is provided to the algorithm.
In the machine learning applications \cite{Emamjomeh-ZadehK17}, the graphs may be constructed in such a way that knowing which objects represent two vertices is sufficient to conclude the distance between them, i.e., a low-complexity distance oracle can be indeed implemented.
This can be seen as a special case of a general approach to achieve distance oracles in practice through the so called distance-labeling schemes (c.f. Gavoille et al. \cite{DBLP:journals/jal/GavoillePPR04} and for practical approaches, c.f. Abraham et al. and Kosowski and Viennot \cite{AbrahamDFGW16,KosowskiV17}).
We finally note that having the exact distances between vertices is crucial for this problem: if the distance oracle is allowed to provide even just a $1$-additive approximation of the exact distance, then each query algorithm needs to perform $\Omega(n)$ queries for some graphs c.f. Deligkas et~al. \cite{DeligkasMS19}. We note that the distance oracle access can be replaced with a multi-source distance computation (e.g. using BFS), at the cost of replacing some of the $\bigo(n)$ factors in the cost functions with $\bigo(m)$. Alternatively, a popular assumption borrowed from computational geometry is that we operate on a metric space with a metric (distance) function given.

\subsection{Our Results and Techniques} \label{sec:our-results}
For a query on a vertex $q$ with a reply $v$, we say that a vertex $u$ is \emph{consistent} with the reply if $q=v=u$, or $q\neq v$ but $v$ lies on a shortest path between $u$ and $q$; the set of all such consistent vertices $u$ is denoted by $N(q,v)$.
Our method is based on a multiplicative weight update (MWU): the algorithm keeps the weights $\vweight(v)$ for all vertices $v$, starting with a uniform assignment.
The weight is representing the likelihood that a vertex is the target, although we point out that formally this is not a probability distribution.
In MWU, the weight of each vertex that is not consistent with a reply is divided by an appropriately fixed constant $\Gamma$ that depends on $\varepsilon=\frac{1}{2}-p$.

To keep the query complexity low, it is required that the queried vertex $q$ fulfills a measure of `centrality' in a graph in the sense that a query to such a central vertex results in an adequate decrease in the total weight.
This is a graph-theoretic analogue of the `central' element comparison in the classical binary search.
Two functions that have been used in the literature \cite{DeligkasMS19,DereniowskiTUW19,Emamjomeh-ZadehK17} to formalize this are 
\[\Phi(v) = \sum_{u \in V} d(u,v) \cdot \vweight(u),\qquad\qquad\text{and}\qquad\qquad\Lambda(v) = \max_{u \in N(v)} \vweight(N(v,u)),\]
where $N(v)$ is the set of neighbors of $v$ in the graph, and $d(u,v)$ is the distance between $u$ and $v$.
For brevity, $\vweight(S)=\sum_{u\in S}\vweight(u)$ for any $S\subseteq V$, and $\vweight=\vweight(V)$.
\begin{definition}
A vertex $q = \arg \min_{v \in V} \Phi(v)$ is called a \emph{median}.
\end{definition}
We note a fundamental bisection property of a median:
\begin{lemma}[c.f. \cite{Emamjomeh-Zadeh:2015aa} section 2] \label{lem:folklore}
If $q$ is a median, then $\Lambda(q) \le \vweight(V)/2$.
\end{lemma}
Such property is key for building efficient binary-search algorithms in graphs, see \cite{DereniowskiTUW19,Emamjomeh-Zadeh:2015aa}: e.g., for the noiseless case, repeatedly querying a median of $X$, where $X \subseteq V$ is the subset of vertices that still can be a target, results in a strategy guaranteeing at most $\log_2 n$ queries.

A disadvantage of using median is that it is computationally costly to find. 
Moreover, using its multiplicative approximation, that is, through a function $\Phi'$ such that $\Phi'(q) = (1 \pm \varepsilon')\Phi(q)$ for any constant $\varepsilon'>0$, blows up the strategy length exponentially \cite{DeligkasMS19} and thus this approach is not suitable.
On the other hand, approximating $\Lambda$-minimizer is feasible, as noted also by \cite{DeligkasMS19}. 

Hence, we work towards a method of efficient median approximation through $\Lambda$ minimization.
We believe that this algorithmic approach is of independent interest and can be used in different graph-theoretic problems.
Interestingly, it turns out that we do not even need a multiplicative approximation of a $\Lambda$-minimizer but we only need that $\Lambda(q)$ is at most roughly half of the total weight.
This is potentially usable in algorithms using generally understood graph bisection.
(For an example of using such balanced separators for somewhat related search with persistent errors see e.g. Boczkowski et~al. \cite{BoczkowskiKR18}.)
Formally, motivated by Lemma~\ref{lem:folklore}, we relax the notion of the median to the following.
\begin{definition} \label{def:Lambda}
We say that a vertex $q^*$ is \emph{$\delta$-close} to a median, for some $\delta>0$, when $\Lambda(q^*) \leq \left( \frac{1}{2}+\delta \right) \cdot \vweight.$
\end{definition}
 
To work-around the fact that $\Lambda$ is not efficient from the algorithmic standpoint, we introduce the following relaxation of $\Phi$:
\[\PhiS(q)=\sum_{v\in S}d(q,v),\]
where $S$ is a random sample of vertices with probability distribution proportional to $\vweight$.
We can now formulate our main contribution in terms of new algorithmic tools:
\begin{description}
 \item[Median approximation.]
 The relaxation of $\Phi$ to $\PhiS$ provides, with high probability, a sufficient approximation of the median vertex in a graph.
\end{description}
We formalize this statement in the following way.
Consider a sample size $s=\frac{8\ln n}{\delta^2}$, where $n$ is the number of vertices of the graph.
This allows us to say how to approximate the median efficiently through a local condition:
\begin{theorem} \label{lem:local-minimum}
Let $q$ be a vertex such that for each $v\in N(q)$ it holds $\PhiS(q) \leq \PhiS(v) + \delta s$.
Then, with high probability at least $1-n^{-3}$, the vertex $q$ is $\delta$-close to a median.
\end{theorem}
As a consequence, we obtain:
\begin{corollary} \label{cor:median-equivalence}
Let $q^* = \arg \min_{v \in V} \PhiS(v)$. 
Then, the vertex $q^*$ is $\delta$-close to a median with high probability at least $1-n^{-3}$.
\end{corollary}

Returning to our search problem, these are enough to both find the right query vertex in each step, keep the strategy length low, and have a centrality measure that is efficient in terms of computational complexity.
This leads us to the following theorem that is based on MWU with some appropriately fixed scaling factor $\Gamma$.
\begin{restatable}{theorem}{complexityA}
\label{thm:complexityA}
Let $p = \frac{1}{2} - \varepsilon$ be the noise parameter for some $0 < \varepsilon \le \frac12$.
There exists an adaptive query algorithm that after asking $\T= \bigo(\frac{\log n}{\varepsilon^2})$ queries returns the target correctly with high probability.
The computational complexity of the algorithm is $\bigo(\frac{n \log n}{\varepsilon})$ per query.
\end{restatable}
The algorithm behind the theorem iterates over the entire vertex set to find a $\PhiS$-minimizer.
We can refine this algorithm for graphs of low maximum degree $\Delta$ and diameter $D$.
For that we use a local search whose direct application requires `visiting' $D\Delta$ vertices to get to a $\PhiS$-minimizer.
However, we introduce two ideas to speed it up.
The first one is adding another approximation layer on top of $\PhiS$: it is not necessary to find the exact $\PhiS$-minimizer but its approximation, which we do as follows.
Whenever the local search moves from one vertex $u$ to its neighbor $v$ and the improvement from $\PhiS(u)$ to $\PhiS(v)$ is sufficiently small, then $v$ will do for the next query.
The second one is to start the local search from the vertex queried in the previous step.
These two ideas combined lead to the second main result.
\begin{restatable}{theorem}{complexityB}
\label{thm:complexityB}
Let $p = \frac{1}{2} - \varepsilon$ for some $0 < \varepsilon \le \frac12$.
There exists an adaptive query algorithm that after asking $\T= \bigo(\frac{\log n}{\varepsilon^2})$ queries returns the target correctly with high probability.
The average computational complexity per query is $\bigo(n+D \Delta \frac{\log n}{\varepsilon^2})$ for graphs with diameter $D$ and maximum degree $\Delta$.
\end{restatable}

\subsection{Related Work} \label{sec:related-work}
Median computation is one of the fundamental ways of finding central vertices of the graph, with huge impact on practical research \cite{Bavelas_1950,Beauchamp_1965,freeman1978centrality,hakimi1964optimum,Sabidussi_1966,tansel1983state}. A significant amount of research has been devoted to efficient algorithms for finding medians of networks \cite{Ostresh_1978,DBLP:conf/dis/TabataNK12,DBLP:journals/ieicet/TabataNK17} or approximating the notion \cite{DBLP:journals/siamjo/CantoneCFP05,DBLP:journals/tcs/Chang12}.  We note the seminal work of Indyk~\cite{DBLP:conf/stoc/Indyk99} which includes $1+\varepsilon$ approximation to $1$-median in time $\bigo(n/\varepsilon^5)$ in metric spaces -- we note that the form of approximation there differs from ours, although the very-high level technique of using random sampling is common. Chechik et~al. in \cite{DBLP:conf/approx/ChechikCK15} use (non-uniform) random sampling to answer queries on sum of distances to the queried vertices in graphs.

We also refer the reader to some recent work on the median computation in median graphs, see Beneteau~et~al. \cite{BeneteauCCV19} and references therein.
More related centrality measures of a graph are discussed in \cite{AbboudGW15,DBLP:conf/soda/AbboudWW16,DBLP:conf/soda/Cabello17} in the context of fine-grained complexity, showing e.g. that efficient computation of a median vertex (in edge-weighted graphs) is equivalent under subcubic reductions to computation of All-Pairs Shortest Paths.

Substantial amount of research has been done on searching in sorted data (i.e., paths), which included investigations for fixed number of errors \cite{Aigner96,RivestMKWS80}, optimal strategies for arbitrary number of errors and various error models, including linearly bounded \cite{DhagatGW92}, prefix-bounded \cite{BorgstromK93} and noisy/probabilistic \cite{Ben-OrH08,KarpK07}.
Also, a lot of research has been done on how different types of queries influence the search process --- see \cite{DaganFGM17} for a recent work and references therein.
The mostly studied comparison queries for paths have been extended to graphs in two ways.
First one is a generalization to partial orders \cite{Ben-AsherFN99,LamY01}, although this does not further generalize well for arbitrary graphs \cite{Dereniowski08}.
It is worth noting that a lot of work has been devoted to the computational complexity of finding error-less strategies \cite{DereniowskiKUZ17,LamY01,MozesOW08}.
The second extension is by using the vertex queries studied in this work, for which much less is known in terms of complexity.
It is worth to mention that the problem becomes equivalent to the vertex ranking problem for trees \cite{Schaffer89}, but not for general graphs (see also \cite{OnakP06}).

Similarly as in the case of the classical binary search, the graph structure guarantees that there always exists a vertex that adequately partitions the search space in the absence of errors \cite{Emamjomeh-Zadeh:2015aa}.
The problem becomes much more challenging as this is no longer the case when errors are present.
A centrality measure that works well for finding the right vertex to be queried is a median used in \cite{DereniowskiTUW19,Emamjomeh-Zadeh:2015aa}.
However, as shown in \cite{DeligkasMS19}, the median is sensitive to approximations in the following way.
When the algorithm decides to query a $(1+\varepsilon')$-approximation $v$ of the median (minimizer of $\Phi'$ which is $1+\varepsilon'$ approximation of $\Phi$), then some graphs require $\bigo(\sqrt{n})$ queries, 
where the approximation is understood as $\Phi(v)\leq(1+\varepsilon')\min_{u\in V}\Phi(u)$.
This results holds for the error-less case.
Furthermore, the authors introduce in~\cite{DeligkasMS19} the potential $\Lambda$ (denoted by $\varGamma$ therein) and prove, also for the error-less case, that it guarantees $\frac{\log_2 n}{1-\log_2(1+\varepsilon)} \approx (1+\varepsilon) \log_2n$ queries, when in each step a $(1+\varepsilon)$-approximation of the $\Lambda$-minimizer is queried.
However, this issue has been considered from a theoretical perspective and no optimization considerations have been made.
In particular, it was left open as to how to reduce the query complexity at an expense of working with such approximations.
This, and the consideration of the noise are two our main improvements with respect to \cite{DeligkasMS19}.
We also stress out that our definition of $\delta$-closeness to a median differs from $(1+\varepsilon)$-approximations in the sense that our definition is much less strict: a vertex $q^*$ that is $\delta$-close to a median may have the property that $\Lambda(q^*)$ significantly deviates from $\min_{u\in V}\Lambda(u)$.

Some complexity considerations have been touched in \cite{Emamjomeh-ZadehK17}, from the perspective of targeting specific machine learning applications, where already the above-mentioned $\Lambda$-minimizer has been used.
To make the statements form that work comparable to our results, we have two distinguish two input size measures that apply.
In \cite{Emamjomeh-ZadehK17}, for a particular application an input consists of a specific machine learning instance, and denote its size by $\tilde{n}$.
In order to find a solution for this instance, a graph $G$ of size $n$ is constructed and an adaptive query algorithm is being run on this graph.
It is assumed that $\log_2n$ is polynomial in $\tilde{n}$.
The diameter $D$ and maximum degree $\Delta$ of $G$ are both assumed in \cite{Emamjomeh-ZadehK17} to be polylogarithmic in $\tilde{n}$.
A local search is used to find a vertex that approximates the $\Lambda$-minimizer.
For that, in each step a sampling is used for the approximation purposes: for each vertex $v$ along the local search, all its neighbors $u$ are tested for finding an approximation $\Lambda$, giving the complexity of $\bigo(D\Delta)=\bigo(m)$, where $m$ is the number of edges of $G$.
It is concluded that the overall complexity of performing a single query is $\bigo(D\Delta \textrm{poly}(\log n,\frac{1}{\varepsilon}))=\bigo(m \cdot \textrm{poly}(\log n,\frac{1}{\varepsilon}))$.

\subsection{Outline}
We proceed in the paper as follows.
Section~\ref{sec:linearly-bounded} provides a `template' strategy in which we simply query a vertex that is $\delta$-close to a median.
The strategy length is there fixed carefully to meet the tail bounds on the error probability.
Then, in Section~\ref{sec:guarantee}, we prove that our sample size is enough to ensure high success probability.
Section~\ref{sec:sampling-complexity} observes that the overall complexity of the algorithm can be reduced by avoiding recasting the entire sample in each step: it is enough to replace only a small fraction of the current sample when going from one step of the strategy to the next.
We then combine these tools to prove our main theorems in Section~\ref{sec:proof}, where for Theorem~\ref{thm:complexityB} we additionally make several observations on speeding-up the classical local search in a graph.

\section{The Generic Strategy} \label{sec:linearly-bounded}
As an intermediate convenient step, we recall the following adversarial error model: given a constant $r$, if the strategy length is $\T$, then it is guaranteed that at most $r\cdot\T$ errors occurred throughout the search (their distribution may be arbitrary). 
We set our parameters as follows: let $\eta = \varepsilon/2$, $r = \frac{1}{2} - \eta$, and assume without loss of generality that $\eta < 1/8$.
Let $\delta = \eta/4$.
With these parameters, we provide Algorithm \textsc{\ref{alg:linearly-bounded-search}} that runs the multiplicative weight update with $\Gamma=\frac{1}{1-4\eta}$ for $\T=\frac{10\log_2n}{\eta^2}$ steps.
Then we prove (cf. Lemma~\ref{lem:linearly-bounded}) that this strategy length is sufficient for correct target detection in this error model.
We write $\vweight_t$ to denote the vertex weight in a step $t$. (So, $\vweight_0$ is the initial uniform weight assignment.)

\begin{figure}[htb]
\begin{center}
\begin{minipage}{.8\linewidth}
\begin{algorithm}[H]
\SetAlgoRefName{LB-Search}
	\caption{Always query a $\delta$-close vertex to a median.}
	\label{alg:linearly-bounded-search}
	$\vweight(v) \gets \frac{1}{n}$ and $\liecount{v}\gets 0$ for each $v\in V$\;
	\For{$\T=10 \frac{\log_2 n}{\eta^2}$ steps}
	{
		Let $q$ be any vertex that is $\delta$-close to a median\label{ln:almost-median}\;
		Query the vertex $q$\;
		\For{each vertex $u$ not compatible with the answer}
		{
			$\vweight(u) \gets \vweight(u) / \Gamma$, where $\Gamma=\frac{1}{1-4\eta}$\;			
			$\liecount{u} \gets \liecount{u}+1$
		}
	}
	\Return the vertex $v$ with the smallest $\liecount{v}$
\end{algorithm}
\end{minipage}
\end{center}
\end{figure}

\begin{lemma} \label{lem:linearly-bounded}
If during the execution of Algorithm \textsc{\ref{alg:linearly-bounded-search}} over total $\T$ queries there were at most $r\cdot \T$ errors, then the algorithm outputs the target.
\end{lemma}
\begin{proof}
If a vertex $v$ at step $t$ satisfies $\vweight_t(v) > (\frac{1}{2} + \delta) \vweight_t$, then we say that $v$ is \emph{heavy} at step $t$.
We aim at proving that the overall weight decreases multiplicatively either by at least $(1-\eta)^2$ or $\frac{\Gamma+1}{2\Gamma}$ per step.
In the absence of a heavy vertex we get the first bound, and it is an immediate consequence of the Equation~\eqref{eq:potdroperror} below.
If we get a heavy vertex at some point, none of these bounds may be true in this particular step (this phenomenon is inherent to the graph query model itself) but we show below that the second one holds in an amortized way (cf. Lemma~\ref{lem:heavy}).
If at step $t$ there is no heavy vertex, then
\begin{equation}
\label{eq:potdroperror}
\vweight_{t+1} \le \left(\frac12+\delta + \frac{\frac12 - \delta}{\Gamma}\right) \vweight_{t} = \left( 1 - 2 \eta + 4 \eta \delta\right)\vweight_t = (1-\eta)^2 \vweight_t.
\end{equation}

Assume otherwise that there is vertex $v$ that is heavy at step $t$. 
\begin{lemma} \label{lem:almost-median-heavy}
If at any step $t$ there is a heavy vertex $v$, then $v$ is the only  $\delta$-close to a median vertex at this step. 
\end{lemma}
\begin{proof}
For any $u\neq v$, we have that $\Lambda(u)\geq\vweight_t(v)>(\frac{1}{2}+\delta)\vweight_t$, i.e., $u$ is not $\delta$-close to a median.
On the other hand, $\Lambda(v)\leq\vweight_t(V\setminus\{v\})<(\frac{1}{2}-\delta)\vweight_t$, i.e., $v$ is $\delta$-close to a median.
\end{proof}

The above lemma implies that if some $v$ is heavy then it will be queried in this particular step.
The next lemma calculates the overall potential drop in a series of steps in which some vertex is heavy.
\begin{lemma}
\label{lem:heavy}
Consider the maximal consecutive segment of steps $\mathcal{I}$ where some $q$ is heavy. That is, we pick $t_1,t_2$ such that $q$ is heavy in all steps $t \in \mathcal{I}=\{t_1, \ldots, t_2-1\}$ and is not heavy in steps $t_1-1$ and $t_2$. Then, $\vweight_{t_2} \le \left(\frac{\Gamma+1}{2\Gamma}\right)^{t_2-t_1} \vweight_{t_1}.$
\end{lemma}
\begin{proof}
First note that, by Lemma~\ref{lem:almost-median-heavy}, $q$ is queried in each step in $\mathcal{I}$.
For a query on $q$, we say that a reply $v$ is a \emph{yes-answer} if $v=q$, and otherwise it is a \emph{no-answer}.
Denote by $a$ and $b$ the number of yes- and no-answers in $\mathcal{I}$, respectively.
Note that $a+b = t_2-t_1$.
Moreover,
\begin{align*}
\vweight_{t_2}(q) &= \left(\frac{1}{\Gamma}\right)^b \vweight_{t_1}(q), &\text{and}&\qquad &\vweight_{t_2}(V\setminus\{q\})\le \left(\frac{1}{\Gamma}\right)^a \vweight_{t_1}(V\setminus\{q\}).
\end{align*}
The vertex $q$ being heavy at $t_1$ implies $\frac{\vweight_{t_1}(q)}{\vweight_{t_1}(V\setminus\{q\})} > \frac{1/2+\delta}{1/2 - \delta}$ and similarly $q$ not being heavy at $t_2$ implies $\frac{\vweight_{t_2}(q)}{\vweight_{t_2}(V\setminus\{q\})} \le \frac{1/2+\delta}{1/2 - \delta}$. Combining the equality and three inequalities above, we obtain $b \ge a$.

We assume without loss of generality that all the yes-answers were given before all the no-answers in the range $\mathcal{I}$. Indeed, we observe that rearranging these answers does not change the state of the algorithm at step $t_2$, and $q$ remains heavy for all of the $\mathcal{I}$. We have then, for the all of the $a$ yes-answers and first $a$ no-answers, the following:
\begin{equation} \label{eq:fist-2a-steps}
\vweight_{t_1+2a} \le  \left(\frac{1}{\Gamma}\right)^a \vweight_{t_1} = \left(\frac{1}{\sqrt{\Gamma}}\right)^{2a} \vweight_{t_1} \le \left(\frac{\Gamma+1}{2\Gamma}\right)^{2a} \vweight_{t_1},
\end{equation}
where the first inequality is due to the fact that each of the $a$ pairs (a pair understood as a no-answer and a yes-answer) scales down each vertex by at least a factor of $\Gamma$, while in the last inequality we have used $1/\sqrt{\Gamma} \le (\Gamma+1)/(2\Gamma)$.

For the remaining $b-a$ steps of $\mathcal{I}$, the weight of $q$ decreases by a factor of $\Gamma$.
Thus, for each $t\in\{t_1+2a,\ldots,t_2-1\}$, using that $q$ is heavy in step $t$:
\[\vweight_{t+1} \leq \vweight_t(V\setminus\{q\})+\frac{\vweight_t(q)}{\Gamma} \leq \vweight_t\left(\frac{1}{2}-\delta\right)+ \frac{\vweight_t}{\Gamma} \left(\frac{1}{2} + \delta\right) \leq\frac{\vweight_t}{2}+\frac{\vweight_t}{2\Gamma}=\frac{\Gamma+1}{2\Gamma}\cdot\vweight_t.\]
Thus, $\vweight_{t_2} \le \left(\frac{\Gamma+1}{2\Gamma}\right)^{b-a} \vweight_{t_1+2a},$
which together with \eqref{eq:fist-2a-steps} completes the proof of Lemma~\ref{lem:heavy}.
\end{proof}

Let $q$ be the target, and $u$ be the output of Algorithm \textsc{\ref{alg:linearly-bounded-search}}. Assume w.l.o.g. that the algorithm run for $\T' \ge \T$ steps.
Since 
\[\T' \ge 10 \frac{\log_2 n}{ \eta^2} \ge \frac{\log_2 n}{r\log_2(1-4\eta) - 2\log_2(1-\eta) },\]
where the inequality follows from $(1/2-\eta) \cdot \log_2(1-4\eta) - 2\log_2(1-\eta) \ge \frac{1}{10} \eta^2$ when $0 \le \eta \le \frac18$, we obtain a bound
\begin{equation} \label{eq:eta-bound}
(1-4\eta)^{r\T'} \ge (1-\eta)^{2\T'} \cdot n.
\end{equation}

We assume that the algorithm outputs an incorrect vertex $u$, and show that it leads to a contradiction. We consider the state of the weights after $\T'$ steps. We consider two cases.
\begin{enumerate}
\item There is no heavy vertex after $\T'$ steps.
We observe that the starting weight satisfies $\vweight_0 = 1$, and by the bound on the number of errors accumulated on target vertex $\target$ (it cannot be more than $r\T'$), we have $\vweight_{\T'} \ge \vweight_{\T'}(\target)+\vweight_{\T'}(u) > \frac{1}{n} \left( \frac{1}{\Gamma} \right)^{r\T'}.$
By Equation~\eqref{eq:potdroperror} and Lemma~\ref{lem:heavy}, we know that every step contributed at least a factor $(1-\eta)^2$ or $(\Gamma+1)/(2\Gamma) = (1-2\eta)$ multiplicatively to the total weight. Thus, by~\eqref{eq:eta-bound},
$\vweight_{\T'} \le (1-\eta)^{2\T'} \vweight_0 \le \frac{1}{n} (1-4\eta)^{r\T'} = \frac{1}{n} \left( \frac{1}{\Gamma} \right)^{r\T'},$
which leads to a contradiction.

\item
Returned vertex $u$ is heavy after $\T'$ steps. We append at the end of the strategy a virtual sequence of $k$ identical query-answers: algorithm queries $u$, and receives an no-answer pointing towards $q$. Here, $k$ is chosen to be minimal such that after $\T'+k$ steps $u$ is no longer heavy (it exists, since each such query increases $\ell_u$ by 1, and leaves $\ell_q$ unchanged). However, at the end of $\T'+k$ round $\ell_u$ is minimal (possibly not necessarily uniquely minimal). We note that appending those $k$ steps did not increase the total number of errors from the answerer, and all of the queries were asked to a heavy vertex $u$. This reduces this case to the previous one, with increased value of $\T'$. \qedhere
\end{enumerate}
\end{proof}

We now transit from the adversarial search to the noisy setting.
This is done by using Algorithm~\textsc{\ref{alg:linearly-bounded-search}} as a black box with $\eta$ being fixed appropriately.
Recall that $p = \frac{1}{2} - \varepsilon$, and we will use the following dependence of $\eta$ on $\varepsilon$ (note that by taking $\eta$ smaller than $\varepsilon$ we accommodate the necessary tail bound in the lemma below, i.e., we ensure that the event of having more than $r\T$ errors is sufficiently unlikely).
\begin{lemma} \label{lem:noise-success-probability}
Run Algorithm~\textsc{\ref{alg:linearly-bounded-search}} with $r = \frac12 - \eta$, where $\eta = \varepsilon/2$.
If an answer to each query was erroneous with probability at most $p$, independently, then the algorithm outputs the target vertex with a high probability of at least $1-n^{-3}$.
\end{lemma}
\begin{proof}
Recall $\T = 10\frac{\log_2 n}{\eta^2}$ in Algorithm~\textsc{\ref{alg:linearly-bounded-search}}.
Denote by $L$ the overall number of errors that have occurred during the execution of the algorithm. The expected number of errors is $p \cdot \T$. By the Hoeffding inequality, $$\text{Pr}[L \ge r \cdot \T ] \le \exp\left(-2 \T (r-p)^2 \right) = \exp( -20 \log_2 n) \le n^{-3}.$$
Thus with high probability number of errors is bounded so that we can apply Lemma~\ref{lem:linearly-bounded} (which in itself gives a deterministic guarantee).
\end{proof}

\section{Sampling Guarantees} \label{sec:guarantee}
To take the `random sampling' counterparts of $\Phi$ and $\Lambda$, consider a $\sample = \{m_1,\ldots,m_s\}$ to be a  multiset of $s$ vertices sampled from $V$ with repetitions, with sampling probabilities $p(v) \sim \vweight(v)$.
That is, for each $m_i$, we have $\Pr(m_i = v) = \frac{\vweight(v)}{\vweight}$ and choices made for $m_i$ are fully independent.
To such an $\sample$ we refer as a \emph{random sample}.
We then define the following potentials
$$\PhiS(v) = \sum_{u \in \sample} d(u,v)\qquad\qquad\text{and}\qquad\qquad\LambdaS(v) = \max_{u \in N(v)} |\sample \cap N(u,v)| ,$$
where the intersection of a multiset $\sample$ with some set $X \subset V$ is defined as a multiset $\sample \cap X= \{ m_i \colon i \in \{ 1 ,\ldots, s \} \wedge m_i \in X \}$.

We note a specific detail regarding these functions -- we will prove and use the fact that in order to find a vertex that is $\delta$-close to a median (a vertex we need to query), it is enough to pick an approximation of the $\PhiS$-minimizer. 
This is slightly counterintuitive, since $\delta$-closeness is defined in terms of $\Lambda$ which has a similar meaning to $\LambdaS$. 
However, the subtlety here is due to a complexity issue --- it is easier to recompute the $\PhiS$ upon updating the sample $\sample$.

We denote $s = |\sample|$ and assume in the rest of the paper that $s = \frac{8\ln n}{\delta^2}$.
In this section we prove that this choice of $s$ is sufficient, and then Section~\ref{sec:sampling-complexity} deals with the complexity issues of the sampling method.
The following is shown in the appendix:
\begin{restatable}{lemma-restate}{samplinglambda}
\label{lem:sampling-lambda}
For any $v$, there is $\frac{\Lambda(v)}{\vweight} \le \frac{\LambdaS(v)}{s}+\delta/2$ with a high probability at least $1 - n^{-3}$.
\end{restatable}
\begin{proof}
Consider any neighbor $u$ of $v$.  Denote by $X_i$ the indicator variable that $m_i \in N(v,u)$. Observe that $|\sample \cap N(v,u)| = \sum_i X_i$ and $\E[X_i]  = \frac{\vweight(N(v,u))}{\vweight}$ and so $\E[ \sum_i X_i ] = s \cdot \frac{\vweight(N(v,u))}{\vweight}$.
By a standard application of Hoeffding bound there is
\[\Pr\left(\E[\sum X_i] - (\sum_i X_i) \ge \frac{s\delta}2\right) \le  e^{-2s(\delta/2)^2} = \frac{1}{n^4}.\] 
So
\[\frac{\vweight(N(v,u))}{\vweight} \le \frac{|\sample \cap N(v,u)|}{s} + \delta/2\] holds with probability at least $1 - n^{-4}$.

Taking union bound over at most $n$ neighbors $v$, we have that with probability at least $1 - n^{-3}$ the following hold
\begin{align*}
\Lambda(v) &= \max_{u \in N(v)} \vweight(N(v,u))\\
&\le \max_{u \in N(v)} \left(\frac{|\sample \cap N(v,u)|}{s}+\delta/2 \right) \cdot \vweight\\
 &= \left(\frac{\LambdaS(v)}{s}+\delta/2\right) \cdot \vweight.
\end{align*}
\end{proof}

\begin{lemma} \label{lem:potdrop}
Let $q$ be a vertex such that $\forall_{v \in N(q)} \PhiS(q) \le \PhiS(v) + \delta s$.
Then, $\LambdaS(q) \le   \frac{s(1 + \delta)}{2}$.
\end{lemma}
\begin{proof}
To see that suppose, towards a contradiction, that $\LambdaS(q) > \frac{s(1 + \delta)}{2}$, i.e. there is $v \in N(q)$, such that $|S \cap N(q, v)| > \frac{s(1 + \delta)}{2}$.
Denote $A^-=\sample\cap N(q,v)$ and $A^+=\sample\setminus A^-$.
Using $\sum_{u\in A^-}d(v,u) = \sum_{u\in A^-}d(q,u) - |A^-|$ and $\sum_{u\in A^+}d(v,u) \leq \sum_{u\in A^+}d(q,u) + |A^+|$, we get
\[\PhiS(v) \leq  |A^+|  - |A^-| + \sum_{u\in\sample}d(q,u) <  \frac{s(1 - \delta)}{2} -  \frac{s(1 + \delta)}{2} + \PhiS(q)  = \PhiS(q) - \delta s \leq \PhiS(v),\]
which yields a contradiction.
\end{proof}

Hence, we can prove Theorem~\ref{lem:local-minimum}:
Combining  Lemma~\ref{lem:potdrop} and Lemma~\ref{lem:sampling-lambda}, $$\Lambda(q) \le \left(\frac{\LambdaS(q)}{s}+\delta/2\right) \cdot \vweight \le \left(\frac12+\delta/2 + \delta/2\right) \cdot \vweight$$
with probability at least $1-n^{-3}$.

\section{Maintaining the Sample} \label{sec:sampling-complexity}
We now discuss the complexity of maintaining the sample $\sample$ upon the vertex weight updates.
Given a sample set $\sample_t$ at step $t$, the next sample $\sample_{t+1}$ is computed by a call to Algorithm~\textsc{\ref{alg:resampling}} below.
\begin{figure}[htb]
\begin{center}
\begin{minipage}{.8\linewidth}
\begin{algorithm}[H]
\SetAlgoRefName{Resampling}
	\caption{Update of the sample after step $t$.}
	\label{alg:resampling}
	
	\ForEach{$x_t\in\sample_t$}
	{
		\If{$x_t$ is consistent with the reply in step $t$}
		{
			$x_{t+1} \gets x_t$
		}
		\Else
		{
			\If{with probability $1/\Gamma$}
			{
				 let $x_{t+1} \gets x_t$
			}
			\Else
			{
				$x_{t+1}$ is drawn randomly from $V$ with distribution proportional to the weights $\vweight_{t+1}$
			}
		}
		Insert $x_{t+1}$ to $S_{t+1}$
	}
\end{algorithm}
\end{minipage}
\end{center}
\end{figure}

The correctness of Algorithm~\textsc{\ref{alg:resampling}} is given by Lemma~\ref{lem:resampling}.
Its proof follows the cases in the pseudo-code to show that both the vertices that remain in the sample and the new ones meet the probability requirements for a random sample.

\begin{restatable}{lemma-restate}{resampling}
 \label{lem:resampling} 
Suppose that in Algorithm~\textsc{\ref{alg:linearly-bounded-search}}, after each weight update the current random sample $\sample$ is recalculated by a call to Algorithm \textsc{\ref{alg:resampling}}.
Then, with high probability at least $1 - n^{-3}$, at most $2\varepsilon |\sample|$ resampling operations occur at each step.
\end{restatable}
\begin{proof}
Recall that after querying a vertex $q$ at step $t$ and receiving an answer $v$, the weights are updated as follows. For each $u \in V$:
\begin{itemize}
\item if $u \in N(q,v)$, then $\vweight_{t+1}(u) \gets \vweight_t(u)$,
\item if $u \not\in N(q,v)$, then $\vweight_{t+1}(u) \gets \vweight_t(u)/\Gamma$,
\end{itemize}
where we recall that $\Gamma = \frac{1}{1-4\eta}$.

Consider a vertex $x_t$. Assume the for every $u \in V$, $\Pr(x_t = u) = \frac{\vweight_t(u)}{\vweight_t}$. We have two cases:
\begin{enumerate}
\item If $u \in N(q,v)$, then 
\begin{align*}\Pr(x_{t+1} = u) &= \Pr(x_t = u) + \Pr(x_t \not\in N(q,v)) \cdot \left(1 - \frac{1}{\Gamma}\right) \cdot \Pr(x_{t+1} \text{ is sampled as } u)\\
 &= \frac{\vweight_t(u)}{\vweight_t} + \frac{\vweight_t(V \setminus N(q,v))}{\vweight_t} \left(1 - \frac{1}{\Gamma}\right) \frac{\vweight_{t+1}(u)}{\vweight_{t+1}}\\
&= \frac{\vweight_{t+1}(u)}{\vweight_t} \left(1+ \frac{\vweight_t(V \setminus N(q,v))}{\vweight_{t+1}} \left(1 - \frac{1}{\Gamma}\right)\right)\\
 & = \frac{\vweight_{t+1}(u)}{\vweight_t} \cdot \frac{\vweight_{t+1} + \vweight_t(V \setminus N(q,v))\left(1 - \frac{1}{\Gamma}\right)}{\vweight_{t+1}}\\
&= \frac{\vweight_{t+1}(u)}{\vweight_t} \cdot \frac{\vweight_{t}(N(q,v)) + \vweight_t(V \setminus N(q,v))\frac{1}{\Gamma}+ \vweight_t(V \setminus N(q,v))\left(1 - \frac{1}{\Gamma}\right)}{\vweight_{t+1}}\\
&= \frac{\vweight_{t+1}(u)}{\vweight_t} \cdot \frac{\vweight_t}{\vweight_{t+1}} = \frac{\vweight_{t+1}(u)}{\vweight_{t+1}}.
\end{align*}
\item Otherwise, if $u \not\in N(q,v)$, then
\begin{align*}
\Pr(x_{t+1} = u) &= \Pr(x_t = u) \cdot \frac{1}{\Gamma} + \Pr(x_t \not\in N(q,v)) \cdot \left(1 - \frac{1}{\Gamma}\right) \cdot \Pr(x_{t+1} \text{ is sampled as } u)\\
&= \frac{\vweight_t(u)}{\vweight_t} \frac{1}{\Gamma} + \frac{\vweight_t(V \setminus N(q,v))}{\vweight_t} \left(1 - \frac{1}{\Gamma}\right) \frac{\vweight_{t+1}(u)}{\vweight_{t+1}}\\
&= \frac{\vweight_{t+1}(u)}{\vweight_t} \left(1+ \frac{\vweight_t(V \setminus N(q,v))}{\vweight_{t+1}} \left(1 - \frac{1}{\Gamma}\right)\right)\\
&= \frac{\vweight_{t+1}(u)}{\vweight_{t+1}}.
\end{align*}
\end{enumerate}
This proves that probabilities for each sample are maintained between steps.

We now bound the actual number of resampling operations necessary. Observe that each element of $\sample_t$ is re-sampled with probability at most $4 \eta = 2 \varepsilon$. Let $K$ denote number of re-sampled vertices. $\E[K] = s \cdot 2\varepsilon$, and then by Chernoff bound $\Pr[K > 4s\varepsilon] \le \exp(- 2s\varepsilon/3) = \exp(- \frac{1024 \ln n}{3\varepsilon}) \le n^{-3}.$
\end{proof}

We comment on the computational complexity of sampling according to a distribution. 
\begin{observation}
\label{obs:sampl}
Sampling $K$ vertices according to distribution $\vweight$ can be done in $\bigo(n + K \log K)$ operations.
\end{observation}
The time for sampling is $\bigo(K \log K)$ for generating sorted list of $K$ real-values picked uniformly at random from $[0,1]$, and $\bigo(n)$ for linear scan of all of the weights from $\vweight$.

\section{Proofs of the Main Theorems} \label{sec:proof}

\complexityA*
\begin{proof}
First, assume without loss of generality that $\frac{\log n}{\varepsilon^2} < n^2$, as otherwise the claimed one-step complexity is $\varepsilon^{-1}n\log n=\Omega(n^2)$.
This can be met by an algorithm that at each step queries a median vertex, see \cite{Emamjomeh-Zadeh:2015aa}. 

Run Algorithm \textsc{\ref{alg:linearly-bounded-search}} that performs $\tau=\bigo(\frac{\log n}{\varepsilon^2})$ queries by Lemma~\ref{lem:noise-success-probability}.
The algorithm maintains a sample  $\sample_t$ at each step $t$ by using Algorithm~\textsc{\ref{alg:resampling}}.
By Corollary~\ref{cor:median-equivalence}, the probability that each step of the algorithm indeed uses a vertex that is $\delta$-close to a median is $1-n^{-3}$.
After each query, the algorithm updates the weights in time $\bigo(n)$, and $\bigo(\varepsilon s)$ vertices are re-sampled by Lemma~\ref{lem:resampling}, for the cost of $\bigo(n + \varepsilon s \log n)$ which is subsumed by other terms. Thus the cost of maintaining the values of $\PhiS$ is $\bigo(\varepsilon s)$ per vertex, or $\bigo(n \varepsilon s)$ in total, which is the dominant cost for the algorithm, with the update being performed as:
\[\PhiS(v) \gets \PhiS(v) - \sum_{u \in \sample_{t+1}\setminus \sample_t} d(u,v) + \sum_{u \in \sample_{t}\setminus \sample_{t+1}} d(u,v).\]
Taking a union bound over all steps, we obtain the high success probability $1 - \bigo(n^{-1})$.
\end{proof}

Now we turn out attention to the proof of Theorem~\ref{thm:complexityB}, where a local search is used.
This is a natural approach that gives an improvement for low-degree low-diameter graphs.
The two `twists' that we add are early termination (see the pseudo-code shown as Algorithm~\textsc{\ref{alg:local-search}}) and resuming from the vertex $v$ that is the output of the previous execution of the local search (which is used in the proof of Theorem~\ref{thm:complexityB}).
The former allows us to directly bound the number of iterations; cf. Observation~\ref{obs:iterations}.

\begin{figure}[htb]
\begin{center}
\begin{minipage}{.8\linewidth}
\begin{algorithm}[H]
\SetAlgoRefName{Local-Search} \caption{Find a local median starting from an input vertex $v$} \label{alg:local-search}
	\While{true}
	{
		$q = \arg \min_{u \in N(v)} \PhiS(u)$\;
		\If{$\PhiS(q) > \PhiS(v) - \delta s$}
		{
			\Return{$v$}
		}
		\Else
		{
			$v = q$
		}
	}
\end{algorithm}
\end{minipage}
\end{center}
\end{figure}

\begin{observation} \label{obs:iterations}
If Algorithm~\textsc{\ref{alg:local-search}} run with an input vertex $v$ returns a vertex $v'$, then the number of iterations is upper-bounded by $1+\frac{\PhiS(v)-\PhiS(v')}{\delta s}$.
\end{observation}

\complexityB*

\begin{proof}
First, w.l.o.g. assume that $\log n /\varepsilon^2 < n$, by the same reasoning as in the proof of Theorem~\ref{thm:complexityA}.

By Lemma~\ref{lem:noise-success-probability}, Algorithm \textsc{\ref{alg:linearly-bounded-search}} that performs $\tau=\bigo(\frac{\log n}{\varepsilon^2})$ queries.
We consider the following modification to Algorithm~\textsc{\ref{alg:linearly-bounded-search}}.
As before, the algorithm updates weights in time $\bigo(n)$ and maintains a sample $\sample_t$ at each step $t$ (by using Algorithm~\textsc{\ref{alg:resampling}}) in time $\bigo(n + \varepsilon s \log n)$ which is subsumed by other terms.
However, instead of choosing a vertex that is $\delta$-close to a median in line~\ref{ln:almost-median}, the updated algorithm runs Algorithm~\textsc{\ref{alg:local-search}} with the previously queried vertex as an input, and sets the output vertex to be the vertex $q$ to be queried.
In other words, at each step $t$, it uses Algorithm~\textsc{\ref{alg:local-search}} with input $v_{t-1}$ which returns $v_t$, and queries $q=v_t$.
The algorithm initializes $v_0$ arbitrarily.

By Lemma~\ref{lem:local-minimum}, $v_t$ is $\delta$-close to a median. 
By Observation~\ref{obs:iterations}, we bound the total number of iterations $K$ done by Algorithm~\textsc{\ref{alg:local-search}} by 
\begin{align*} K &\le \sum_{t=0}^{\tau-1} \left(1 + \frac{\PhiS_{t+1}(v_t) - \PhiS_{t+1}(v_{t+1})}{\delta s} \right)\\ 
&= \tau + \frac{\PhiS_1(v_0) + \sum_{t=1}^{\tau-1} (\PhiS_{t+1}(v_t) -\PhiS_t(v_t)) - \PhiS_{\tau}(v_{\tau})}{\delta s}\\
&\le \tau + \frac{s D + 2\tau s \varepsilon D}{\delta s} = \bigo(D \tau),
\end{align*}
where we used that $\PhiS_{t+1}(v_t) -\PhiS_t(v_t) \le 2s\varepsilon D$ holds with high probability by Lemma~\ref{lem:resampling}.
Each iteration in Algorithm~\textsc{\ref{alg:local-search}} has complexity $\bigo(\Delta s)$ making the total complexity of the algorithm to be $\bigo(\tau (n + D \Delta s))$.
\end{proof}

\section{Open Problems} \label{sec:summary}

Having an algorithm that keeps an optimal query complexity and obtains a low computational complexity, one can ask what are the possible tradeoffs between the two?
Another question is how much further the computational complexity can be decreased?
Also, are there any possible lower bounds that can reveal the limits of what is not achievable in the context of these problems?
Regarding the centrality measures we consider, we propose an efficient median approximation.
Motivated by this, another question is what are other possible vertex-functions that may allow for further improvements, e.g. in the complexity?

\clearpage
\bibliography{bib}
\end{document}